\documentclass[11pt,fullpage]{article}

\newtheorem{definition}{Definition}
\newtheorem{prop}{Proposition}
\newtheorem{theorem}{Theorem}

\def\proofskip{\vskip 4pt plus 1pt minus 1pt}

\newenvironment{proof}{{\noindent \it Proof.\/}~}{\proofskip}

\newenvironment{notation}{{\vspace{8pt}\noindent\bf Notation}}{}

\def\la{\leftarrow}
\def\ra{\rightarrow}
\def\ex{\exists}

\def\vdashm{\vdash^{ ^{ \hskip -4pt \mbox{\tiny  M}}}}
\def\implies{\Rightarrow}

\newbox\tempa
\newbox\tempb
\newdimen\tempc
\def\mud#1{\hfil $\displaystyle{\mathstrut #1}$\hfil}
\def\rig#1{\hfil $\displaystyle{#1}$}
\def\irulehelp#1#2#3{\setbox\tempa=\hbox{$\displaystyle{\mathstrut #2}$}
		        \setbox\tempb=\vbox{\halign{##\cr
	\mud{#1}\cr
	\noalign{\vskip\the\lineskip} 
	\noalign{\hrule height 0pt} 
	\rig{\vbox to 0pt{\vss\hbox to 0pt{${\; #3}$\hss}\vss}}\cr
	\noalign{\hrule} 
	\noalign{\vskip\the\lineskip} 
	\mud{\copy\tempa}\cr}} 
		      \tempc=\wd\tempb
		      \advance\tempc by \wd\tempa
		      \divide\tempc by 2 }
\def\irule#1#2#3{{\irulehelp{#1}{#2}{#3}
		     \hbox to \wd\tempa{\hss \box\tempb \hss}}}

\input amssym.def
\input amssym
\def\rhd{\vartriangleright}
\def\conv{=_{\beta\eta}}
 
\font\cur=cmbsy10
\def\Ru{\mbox{\cur R}}

\begin{document}

\thispagestyle{empty}

\title{On the Definition of the Eta-long  Normal Form\\
in Type Systems of the Cube}

\author{Gilles Dowek~~~~~G\'{e}rard Huet~~~~~Benjamin Werner
\thanks{INRIA-Rocquencourt, B.P. 105, 78153
Le Chesnay Cedex, France. Gilles.Dowek@inria.fr, Gerard.Huet@inria.fr,
Benjamin.Werner@inria.fr}}
\date{}
\maketitle

\thispagestyle{empty}
\begin{abstract}
The smallest transitive relation $<$ on well-typed normal terms such that 
if $t$ is a strict subterm of $u$ then $t < u$ and 
if $T$ is the normal form of the type of $t$ and the term $t$ is not a sort
then $T < t$
is well-founded in the type systems of the cube. 
Thus every term admits a $\eta$-long normal form.
\end{abstract}

\section*{Introduction}

In this paper we prove that the smallest transitive relation $<$ on
well-typed normal terms such that 
\begin{itemize}
\item 
if $t$ is a strict subterm of $u$ then $t < u$,
\item
if $T$ is the normal form of the type of $t$ and the term $t$ is not a sort
then $T < t$
\end{itemize}
is well-founded in the type systems of the cube \cite{Barendregt}.
This result is proved using the notion of {\it marked terms} introduced by
de Vrijer \cite{deVrijer}. A motivation for this theorem is to
define the $\eta$-long form of a normal term in these type systems.

In simply typed $\lambda$-calculus, to define the $\eta$-long form of a normal
term we first define the $\eta$-long form of a variable $x$ of type
$P_{1} \ra ... \ra P_{n} \ra P$ ($P$ atomic) as the term
$[y_{1}:P_{1}] ... [y_{n}:P_{n}](x~y'_{1}~...~y'_{n})$
where $y'_{i}$ is the $\eta$-long form of the variable $y_{i}$ of type $P_{i}$.
Then we define the $\eta$-long form of a normal term $t$ well-typed of type
$P_{1} \ra ... \ra P_{n} \ra P$ ($P$ atomic) as
\begin{itemize}
\item
if $t = [x:U]u$ then its $\eta$-long form is the term $[x:U]u'$ where $u'$
is the $\eta$-long form of $u$,
\item
if $t = (x~c_{1}~...~c_{p})$ then its $\eta$-long form is the term 
$[y_{1}:P_{1}] ... [y_{n}:P_{n}](x~c'_{1}~...~c'_{p}~y'_{1}~...~y'_{n})$
where $c'_{i}$ is the $\eta$-long form of the term $c_{i}$,
and $y'_{i}$ the $\eta$-long form of the variable $y_{i}$.

\end{itemize}
The definition of the $\eta$-long form of a variable is by induction over the
structure of its type, and the definition of the $\eta$-long form of a normal
term is by induction over the structure of the term itself.
The $\eta$-long form appeared in \cite{Pietrzykowski} under the name of
{\sl long reduced form} and in \cite{Huet75} under the name of
{\sl $\eta$-normal form}, and was further
investigated in \cite{Huet76}, under the name of {\sl extensional form}.

In systems with dependent types the corresponding  definition is more
complicated.
First when $t = [x:U]u$ we have to take also the $\eta$-long form of the 
term $U$ and when $t = (x~c_{1}~...~c_{p})$ we have to take also the 
$\eta$-long form of the terms $P_{1}, ..., P_{n}$. So the well-foundedness 
of this
definition is not so obvious, indeed $P_{i}$ is not a subterm of $t$, but a 
subterm {\it of its type}. We prove the well-foundedness of this definition 
using the well-foundedness of the relation $<$. 

Besides the definition of $\eta$-long form, this well-foundedness
result has been used in
\cite{complete} to prove the completeness of the resolution method in the 
systems of the cube and in \cite{matching} to prove the decidability of second
order matching in these systems.

\section{The Cube of Typed $\lambda$-calculi}

\begin{definition}[Term]

The set of terms is inductively defined as 
$$T~::=~Prop~|~Type~|~x~|~(T~T)~|~[x:T]T~|~(x:T)T$$
\end{definition}

The symbols $Prop$ and $Type$ are called {\it sorts}, the terms $x$
are called {\it variables}, the terms $(T~T')$ {\it applications}, the
terms $[x:T]T'$ {\it abstractions} (they are sometimes written
$\lambda x:T.T'$) and the terms $(x:T)T'$ {\it products} (they are
often written $\Pi x:T.T'$).  Products are the type of functions and
the notation $T \ra T'$ is used for $(x:T)T'$ when $x$ does not occur
free in $T'$. In this paper we ignore variable renaming problems. A
rigorous presentation would use de Bruijn indices.  We write
$(a~b_{1}~...~b_{n})$ for $(~...~(a~b_{1})~...~b_{n})$.

\begin{notation}
If $t$ and $u$ are two terms and $x$ is a variable,
we write $t[x \la u]$ for the term obtained by substituting $u$ for
$x$ in $t$.  We write $a =_{\beta \eta} b$ (resp. $a =_{\beta} b$, $a
=_{\eta} b$) when the terms $a$ and $b$ are $\beta \eta$-convertible
(resp.  $\beta$-convertible, $\eta$-convertible).
We write $t \rhd u$ (respectively $\rhd_{\beta}$, $\rhd_{\eta}$) when
$t$ reduce in one step of $\beta$ or $\eta$-reduction 
(resp. $\beta$-reduction, $\eta$-reduction)
to $u$,
$t \rhd^{*} u$ (resp. 
$t \rhd^{*}_{\beta} u$, $t \rhd^{*}_{\eta} u$)
when $t$ reduces in an arbitrary
number of steps to $u$ and $t \rhd^{+} u$ 
(resp. 
$t \rhd^{+}_{\beta} u$, $t \rhd^{+}_{\eta} u$)
when $t$ reduces in
at least one step to $u$. 
\end{notation}

\begin{definition}[Context]

A {\it context} $\Gamma$ is a list of pairs $<x,T>$ (written $x:T$) where $x$ 
is a variable and $T$ a term. The term $T$ is called the {\it type} of $x$ in
$\Gamma$.
\end{definition}
We write $[x_{1}:T_{1}; ...; x_{n}:T_{n}]$ for the context with elements 
$x_{1}:T_{1}, ..., x_{n}:T_{n}$ and $\Gamma_{1}\Gamma_{2}$ for the 
concatenation of the contexts $\Gamma_{1}$ and $\Gamma_{2}$.

\begin{definition}[Typing rules]

We define inductively two judgements: {\em $\Gamma$ is well-formed}
(written $\vdash \Gamma$) and {\em $t$ has type $T$ in $\Gamma$}
(written $\Gamma \vdash t:T$) where $\Gamma$ is a context and $t$ and
$T$ are terms. These judgements are indexed by the parameter $\Ru$
which is a set of pairs of sorts that contains the pair $<Prop,Prop>$.
\end{definition}
\begin{center}
$$\irule{}
        {\vdash [~]}
        {}$$
$$\irule{\Gamma \vdash T:s~~x\notin\Gamma}
        {\vdash \Gamma[x:T]}
        {s \in \{Prop, Type\}}$$
$$\irule{\vdash \Gamma}
        {\Gamma \vdash Prop:Type}
        {}$$
$$\irule{\vdash \Gamma~~x:T \in \Gamma} 
        {\Gamma \vdash x:T}
        {}$$
$$\irule{\Gamma \vdash T:s~~\Gamma[x:T] \vdash U:s'}
        {\Gamma \vdash (x:T)U:s'}
        {<s,s'> \in \Ru}$$
$$\irule{\Gamma \vdash (x:T)U:s~~\Gamma[x:T] \vdash t:U}
        {\Gamma \vdash [x:T]t:(x:T)U}
        {s \in \{Prop, Type\}}$$
$$\irule{\Gamma \vdash t:(x:T)U~~\Gamma \vdash u:T}
        {\Gamma \vdash (t~u):U[x \la u]}
        {}$$
$$\irule{\Gamma \vdash t:T~~\Gamma \vdash U:s~~T =_{\beta \eta} U}
        {\Gamma \vdash t:U}
        {s \in \{Prop, Type\}}$$
\end{center}

There are eight choices for the set $\Ru$ defining
eight calculi.
Examples of such calculi are the simply typed $\lambda$-calculus 
($\Ru = \{<Prop, Prop>\}$),
the $\lambda \Pi$-calculus \cite{HHP} ($\Ru = \{<Prop, Prop>, <Prop,Type>\}$),
the system $F$ \cite{Girard}
($\Ru = \{<Prop, Prop>, <Type,Prop>\}$), the system
$F\omega$ \cite{Girard} ($\Ru = \{<Prop, Prop>, <Type,Prop>, <Type,Type>\}$)
and the Calculus of Constructions \cite{Coquand85,CoqHue}
($\Ru = \{<Prop, Prop>, <Prop,Type>, <Type,Prop>, <Type,Type>\}$).

\begin{definition}[Well-typed term]

A term $t$ is said to be {\it well-typed} in a context $\Gamma$ iff 
there exists a term $T$ such that $\Gamma \vdash t:T$.
\end{definition}

Notice that the term $Type$ is not well-typed.

\medbreak
\noindent
The following facts hold for every type system of the cube. The proofs
are given in~\cite{Geuvers93,Werner}. 

\begin{prop}[Substitution]
If $\Gamma [x:U] \Gamma' \vdash t:T$ and $\Gamma \vdash u:U$ then
$\Gamma \Gamma' [x \la u] \vdash t[x \la u]:T[x \la u]$. 
\end{prop} 

\begin{prop}
If $\Gamma \vdash t:T$ then either $T = Type$ or there exists a sort
$s$ such that $\Gamma \vdash T:s$.
Moreover if $t$ is a variable, an application or an abstraction (but neither
a sort nor a product) then $T \neq Type$.
\end{prop}

\begin{prop}[Stripping Lemma]
Let $\Gamma\vdash t:T$ be a derivable judgement. The following
holds:
\begin{eqnarray*}
t=Prop & \implies & T=Type \\
t=x & \implies & \ex U. (x:U\in\Gamma\wedge T\conv U) \\
t=(x:A)B & \implies & \exists s_1, s_2. \Gamma\vdash A:s_1\wedge\Gamma[x:A]\vdash
B:s_2\wedge(s_1,s_2)\in\Ru\wedge T=s_2\\
t=(u~v)&\implies&
 \exists A,B.(\Gamma\vdash u:(x:A)B\wedge\Gamma\vdash v:A\wedge
T\conv B[x \la v])\\
t = [x:A]u& \implies &
 \exists B. \exists s.
(\Gamma\vdash (x:A)B:s\wedge
 \Gamma[x:A]\vdash u:B \wedge
 T\conv (x:A)B )      
\end{eqnarray*}
\end{prop}

\begin{prop}[Type Uniqueness]
A well-typed term has a unique type up to conversion.
\end{prop}

\begin{definition}[Atomic Term]
A term is said to be {\it atomic} if it has the form $(h~c_{1}~...~c_{n})$
where $h$ is a variable or a sort (in this last case with $n=0$). 
The symbol $h$ is called the {\it head} of this term.
\end{definition}

\begin{prop}
Let $T$ be a well-typed $\beta$-normal term of type $s$ for some sort 
$s$, $T$ can be written in a unique way 
$T = (x_{1}:P_{1}) ... (x_{n}:P_{n})P$ with $P$ atomic. Moreover if 
$s = Type$ then $P = Prop$. 
\end{prop}

\begin{prop}[Strengthening]
If $\Gamma[x:U]\Gamma'\vdash t:T$ and $x$ does not occur free in
$\Gamma'$, $t$ and $T$, then $\Gamma\Gamma'\vdash t:T$.
\end{prop}

\begin{prop}[Subject reduction]
If $\Gamma \vdash t:T$ and $t \rhd^{*} u$ then $\Gamma \vdash u:T$.
\end{prop}

\begin{prop}[Normalization]
The relation $\rhd$ is strongly normalizing.
\end{prop}

\begin{prop}[Confluence] 
The relation $\rhd$ is confluent.
\end{prop}

\begin{prop}[Church-Rosser]
If $\Gamma\vdash t_1:T$, $\Gamma\vdash t_2:T$ and $t_1\conv t_2$, then
there exists $t$ such that $t_1\rhd^{*}t$ and $t_2\rhd^{*}t$ (and so
$\Gamma\vdash t:T$).
\end{prop}

We will also need the following technical lemma. It is implicitly
present in the Church-Rosser proofs of~\cite{Geuvers93,Werner} but not
stated in this exact form. It basically states that $\eta$-reductions
do not erase free-variables on well-formed $\beta$-normal terms.

\begin{prop}\label{tec1}
Let $\Gamma[x:U]\Gamma'\vdash t:T$ be a derivable judgement, with $t$
$\beta$-normal. Let $t'$ be its $\beta\eta$-normal form. If
$\Gamma\Gamma'\vdash t':T$ (in other words, $x$ does not occur free in
$\Gamma'$, $t'$ and $T$) then $x$ does not occur free in $t$ either
(and hence, by strengthening, $\Gamma\Gamma'\vdash t:T$).
\end{prop}

\begin{proof}
On $\beta$-normal terms, such as $t$, $\eta$-reduction does not create
$\beta$-redexes. Thus we have $t\rhd^{*}_{\eta} t'$. We prove the
proposition by induction over the structure of $t$; since $t$ is
$\beta$-normal, the cases are:
\begin{itemize}
\item $t=s$; trivial.
\item $t=(z:A)B$; in which case $t'=(z:A')B'$ with $A\rhd^{*}_{\eta}A'$
and $B\rhd_{\eta}^{*}B'$. The stripping lemma allows us to simply
apply the induction hypothesis to the subterms (first to $A$, then to
$B$ in the context $\Gamma[x:U]\Gamma'[z:A]$).
\item $t=(y~u_1~\dots~u_n)$; which is similar to the case above. We
simply have to reason by induction over $n$ in order to first build
the respective types of $u_1$\dots$u_n$ in $\Gamma\Gamma'$.
\item $t = [z:A]B$; then we have two possibilities:
\begin{itemize}
\item $t' = [z:A']B'$ with $A\rhd^{*}_{\eta}A'$ and $B\rhd_{\eta}^{*}B'$.
This case is similar to $t=(z:A)B$.
\item $t'=B'$ with $B\rhd_{\eta}^{*}(B''~z)$ and
$B''\rhd^{*}_{\eta}B'$ with $x$ not free in $B'$. 
By the stripping lemma and Church-Rosser property, we know that
$T\rhd^{*}(z:A_0)C$ with $A\conv
A_0$. Since $x$ does not occur free in $A_0$, neither does it in its
normal form, which is also the normal form of $A$. 
By induction hypothesis, $x$ does not occur free in $A$. 
Furthermore, the $\beta\eta$-normal form of $B$ is $(B'~z)$, in which
$x$ does not occur free. Since $C$ is a type of $B$ in which $x$ does
not occur free either, we may again apply the induction hypothesis
to the derivable judgement
$$\Gamma[x:U]\Gamma'[z:A]\vdash B:C$$
in order to verify that $x$ does not occur free in $B$.
\end{itemize}
\end{itemize}
\end{proof}

In order to be able to talk about the type of a subterm, we define the
usual notion of subterm, but keeping track of the context in which
they are well-typed: we write $t_\Gamma$ for a pair composed by a
context $\Gamma$ and a {\it $\beta$-normal} term $t$ 
known to be well-typed in $\Gamma$.

\begin{definition}[Subterm]

Let $t_{\Gamma}$ be such a pair; we define by induction over the structure of 
$t_{\Gamma}$ the set $Sub(t_{\Gamma})$ of {\it strict subterms} of 
$t_{\Gamma}$:
\begin{itemize}
\item{if $t_{\Gamma}$ is a sort or a variable then $Sub(t_{\Gamma}) = \{\}$,}
\item{if $t_{\Gamma}$ is an application $t = (u~v)$ then
$Sub(t_{\Gamma}) = \{u_{\Gamma}, v_{\Gamma}\} \cup Sub(u_{\Gamma}) \cup 
Sub(v_{\Gamma})$,}
\item{if $t_{\Gamma}$ is an abstraction $t = [x:P]u$ then
$Sub(t_{\Gamma}) = \{P_{\Gamma},u_{\Gamma[x:P]}\} \cup Sub(P_{\Gamma}) \cup 
Sub(u_{\Gamma[x:P]})$,}
\item{if $t_{\Gamma}$ is a product $t = (x:P)u$ then
$Sub(t_{\Gamma}) = \{P_{\Gamma},u_{\Gamma[x:P]}\} \cup Sub(P_{\Gamma}) \cup
Sub(u_{\Gamma[x:P]})$.}
\end{itemize}
\end{definition}

Notice that if the term $t$ is well-typed in $\Gamma$ and 
$u_{\Delta} \in Sub(t_{\Gamma})$ then $u$ is well-typed in $\Delta$.

\begin{definition}[The Relation $<$]

Let $<$ be the smallest transitive relation defined on $\beta\eta$-normal
well-typed terms labeled by their contexts such that
\begin{itemize}
\item 
if neither $t$ nor $u$ is a sort and $t_{\Gamma}$ is a strict subterm of 
$u_{\Delta}$ then $t_{\Gamma} < u_{\Delta}$,
\item
if neither $t$ nor $T$ is a sort, and $T$ is the (unique) normal 
form of the type of $t$ in $\Gamma$ then $T_{\Gamma} < t_{\Gamma}$.
\end{itemize}
\end{definition}

\section{Marked Terms}

\subsection{Definition}

The main idea in this proof is to consider a new syntax for type systems
such that the type of a well-typed term $t$ is a subterm of $t$. So we define 
a syntax for type systems where each term is marked with its type. In fact, we
only need to mark variables, applications and abstractions, but neither sorts 
nor products.

\begin{definition}[Marked Terms]
$$T~::=~Prop~|~Type~|~x^{T}~|~(T~T)^{T}~|~([x:T]T)^{T}~|~(x:T)T$$
\end{definition}

\begin{definition}[Marked context]
A {\it marked context} is a list of pairs $<x,T>$ (written $x:T$) where $x$ 
a variable and $T$ a marked term.
\end{definition}

\begin{notation}
Let $t$ and $u$ be marked terms and $x$ a variable. We write $t[x \la u]$
for the term obtained by substituting $u$ for $x$ in $t$, since
the variable $x$ may also occur in the marks we have to substitute both in the 
term and in the marks.
Free variables are defined as usual, here also they may occur in the marks
too.
$\beta$ and $\eta$-reductions are defined on marked terms in
the same way as they are on unmarked terms. We write $t \rhd u$ when $t$
reduces in one step of $\beta$ or $\eta$-reduction to $u$.
The contracted redex may be either in the term or in the marks.
We write $t \rhd_{\beta} u$ if the contracted redex is a
$\beta$-redex and $t \rhd_{\eta} u$ if the contracted redex is an
$\eta$-redex.
A marked term $t$ is said to be normal (resp. $\beta$-normal,
$\eta$-normal) if it contains no redex (resp. $\beta$-redex, $\eta$-redex).
\end{notation}

\begin{definition}[Contents]

Let $t$ be a marked term, the {\it contents} of $t$ is the unmarked term 
$t^{\#}$ defined by induction over the structure of $t$:
\begin{itemize}
\item{if $t$ is a sort then $t^{\#} = t$,}
\item{if $t = x^{T}$ then $t^{\#} = x$,}
\item{if $t = (u~v)^{T}$ then 
$t^{\#} = (u^{\#}~v^{\#})$,}
\item{if $t = ([x:P]u)^{T}$ then 
$t^{\#} = [x:P^{\#}]u^{\#}$,}
\item{if $t = (x:P)U$ then 
$t^{\#} = (x:P^{\#})U^{\#}$.}
\end{itemize}

Let $\Gamma = [x_{1}:P_{1}; ...; x_{n}:P_{n}]$ be a marked context, the 
contents of $\Gamma$ is the context
$$\Gamma^{\#} = [x_{1}:P_{1}^{\#}; ...; x_{n}:P_{n}^{\#}]$$
\end{definition}

\begin{definition}[Conversion on marked terms]
Let $t_1$ and $t_2$ be two {\it marked} terms. We
say that $t_1\conv t_2$ (respectively $t_1=_\beta t_2$) if and only
if $t_1^{\#}\conv t_2^{\#}$ (respectively $t_1^{\#} =_\beta t_2^{\#}$)
\end{definition}

Notice that the convertibility relation used on marked terms is
defined as $\beta\eta$-conversion on the underlying unmarked terms,
whatever that marks are.

\begin{definition}[Typing rules]

We define inductively two judgements: {\it $\Gamma$ is well-formed} 
(written $\vdashm \Gamma$)
and 
{\it $t$ has type $T$ in $\Gamma$} ($\Gamma \vdashm t:T$) where $\Gamma$ is
a marked context and $t$ and $T$ are marked terms.
\end{definition}
\begin{center}
$$\irule{}
        {\vdashm [~]}
        {}$$
$$\irule{\Gamma \vdashm T:s~~x\notin\Gamma} 
        {\vdashm \Gamma[x:T]}
        {s \in \{Prop, Type\}}$$
$$\irule{\vdashm \Gamma}
        {\Gamma \vdashm Prop:Type}
        {}$$
$$\irule{\vdashm \Gamma~~x:T \in \Gamma} 
        {\Gamma \vdashm x^{T}:T}
        {}$$
$$\irule{\Gamma \vdashm T:s~~\Gamma[x:T] \vdashm U:s'}
        {\Gamma \vdashm (x:T)U:s'}
        {<s,s'> \in \Ru}$$
$$\irule{\Gamma \vdashm (x:T)U:s~~\Gamma[x:T] \vdashm t:U} 
        {\Gamma \vdashm ([x:T]t)^{(x:T)U}:(x:T)U}
        {s \in \{Prop, Type\}}$$
$$\irule{\Gamma \vdashm t:(x:T)U~~\Gamma \vdashm u:T}
        {\Gamma \vdashm (t~u)^{U[x \la u]}:U[x \la u]}
        {}$$
$$\irule{\Gamma \vdashm t:T~~\Gamma \vdashm U:s~~T =_{\beta \eta} U}
        {\Gamma \vdashm t:U}
        {s \in \{Prop, Type\}}$$
$$\irule{\Gamma \vdashm x^{V}:T~~\Gamma \vdashm W:s~~V =_{\beta \eta}  W}
        {\Gamma \vdashm x^{W}:T}
        {s \in \{Prop, Type\}}$$
$$\irule{\Gamma \vdashm ([x:T]t)^{V}:U~~\Gamma \vdashm W:s~~V =_{\beta \eta} W}
        {\Gamma \vdashm ([x:T]t)^{W}:U}
        {s \in \{Prop, Type\}}$$
$$\irule{\Gamma \vdashm (t~u)^{V}:U~~\Gamma \vdashm W:s~~V =_{\beta \eta} W}
        {\Gamma \vdashm (t~u)^{W}:U}
        {s \in \{Prop, Type\}}$$
\end{center}
\normalsize

The three last rules permit to perform conversions in marks, in 
the same way as the rule above permit to perform conversion in the
type of a term.

\begin{definition}[Well-typed marked term]

A marked term $t$ is said to be {\it well-typed} in a marked context $\Gamma$ 
if and only if there exists a marked term $T$ such that $\Gamma \vdashm t:T$.
\end{definition}

\begin{prop}
If $t$ and $u$ are marked terms and $t \rhd u$
then either $t^{\#} = u^{\#}$ (if the redex occurs in the marks) or
$t^{\#} \rhd u^{\#}$.
\end{prop}

\begin{prop}\label{beta}
If $t$ is a marked term such that 
$t^{\#} \rhd_{\beta} u$ then there exists a marked term $t'$ such that
$t \rhd_{\beta} t'$ and $t'^{\#} = u$. 
\end{prop}

Notice that this proposition is false for $\eta$-reduction, even for
well-typed terms; indeed the term $t =
([x:T^{Prop}](y^{(([z:T^{Prop}](T^{Prop} \ra T^{Prop}))^{T^{Prop} \ra
Prop}~x^{T^{Prop}})^{Prop}}~x^{T^{Prop}})^{T^{Prop}}) ^{T^{Prop} \ra
T^{Prop}}$ is not an $\eta$-redex, but $t^{\#} = [x:T](y~x)$ is one.

\begin{prop}
Let $\Gamma$ be a marked context and $t$ and $T$ be marked terms
such that $\Gamma \vdashm t:T$. Then $\Gamma^{\#} \vdash t^{\#}:T^{\#}$.
\end{prop}

\begin{proof}
By induction on the length of the derivation of $\Gamma \vdashm t:T$.
\end{proof}

\subsection{Basic properties}

We now go to proving the basic properties of the system with
marked terms. For several of the following properties, the proofs are
very similar to what is done for the usual formulation using unmarked
terms \cite{Geuvers93,Werner}. In these cases, we do not detail the proof.

\begin{prop}[Uniqueness of Product Formation]
If $(x:T_{1})T_{2} =_{\beta\eta} (x:U_{1})U_{2}$ then 
$T_{1} =_{\beta\eta} U_{1}$ and $T_{2} =_{\beta\eta} U_{2}$.
If $s_{1}$ and $s_{2}$ are two sorts such that $s_{1} =_{\beta\eta}
s_{2}$ then $s_{1} = s_{2}$. 
\end{prop}

\begin{proof}
By definition $(x:T_{1})T_{2} =_{\beta\eta} (x:U_{1})U_{2}$ if and
only if $(x:T_{1}^{\#})T_{2}^{\#} =_{\beta\eta}
(x:U_{1}^{\#})U_{2}^{\#}$ i.e.  $T_{1}^{\#} =_{\beta\eta} U_{1}^{\#}$
and $T_{2}^{\#} =_{\beta\eta} U_{2}^{\#}$, i.e.  $T_{1} =_{\beta\eta}
U_{1}$ and $T_{2} =_{\beta\eta} U_{2}$.
\end{proof}

\begin{prop}[Substitution]
If $\Gamma [x:U] \Gamma' \vdashm t:T$ and
$\Gamma \vdashm u:U$ then\\
$\Gamma \Gamma' [x \la u] \vdashm t[x \la u]:T[x \la u]$.
\end{prop}

\begin{proof}
By induction on the length of the derivation of $\Gamma [x:U] \Gamma' 
\vdashm t:T$.
\end{proof}

\begin{prop}[Weakening]
If $\Gamma \Gamma' \vdashm t:T$, $\Gamma \vdashm U:s$ 
and $x$ is a variable not declared in $\Gamma \Gamma'$
then $\Gamma [x:U] \Gamma' \vdashm t:T$. 
\end{prop}

\begin{proof}
By induction on the length of the derivation of $\Gamma \Gamma' \vdashm t:T$.
\end{proof}

\begin{prop}
If $\Gamma \vdashm t:T$ then either $T = Type$ or there exists a sort
$s$ such that $\Gamma \vdashm T:s$.
Moreover if $t$ is a variable, an application or an abstraction (but neither
a sort nor a product) then $T \neq Type$.
\end{prop}

\begin{proof}
By induction on the length of the derivation of
$\Gamma \vdashm t:T$. For induction loading, we simultaneously
prove that if $\Gamma = \Delta[z:V]\Delta'$, then $\Delta\vdashm V:s$.
\end{proof}

\begin{prop}[Stripping]~

\begin{itemize}
\item
If $\Gamma \vdashm Prop:T$ then $T = Type$,
\item
if $\Gamma \vdashm x^{T}:U$ then there exists a sort $s$ and a declaration 
$x:V$ 
in $\Gamma$ such that $\Gamma \vdashm T:s$ and 
$T =_{\beta\eta} U =_{\beta\eta} V$,
\item
if $\Gamma \vdashm (x:T_{1})T_{2}:U$ then there exists two sorts $s_{1}, 
s_{2}$ such that $U = s_{2}$,
$\Gamma \vdashm T_{1}:s_{1}$, $\Gamma [x:T_{1}] \vdashm T_{2}:s_{2}$ and 
$<s_{1}, s_{2}> \in \Ru$,
\item
if $\Gamma \vdashm ([x:T]t)^{U}:V$ then there exists a term $W$ and
two sorts $s_{1}$ and $s_{2}$ such that 
$U =_{\beta\eta} V =_{\beta\eta} (x:T)W$,
$\Gamma \vdashm U:s_{2}$,
 $\Gamma \vdashm T:s_{1}$,
$\Gamma [x:T] \vdashm t:W$, $\Gamma [x:T] \vdashm W:s_{2}$ and 
$<s_{1}, s_{2}> \in \Ru$,
\item
if $\Gamma \vdashm (t~u)^{T}:U$ then there exists two terms $V$ and $W$
and a sort $s$
such that
$T =_{\beta\eta} U =_{\beta\eta} W[x \la u]$, $\Gamma \vdashm t:(x:V)W$,
$\Gamma \vdashm u:V$ and $\Gamma \vdashm T:s$. 
\end{itemize}
\end{prop}

\begin{proof} 
By induction over the length of the derivation.
\end{proof}

\begin{prop}[Type Uniqueness]
A well-typed term has a unique type up to conversion.
\end{prop}

\begin{definition}[Atomic Term]
A marked term is said to be {\it atomic} if it has the form 
$$(~...~(h~c_{1})^{T_{1}}~...~c_{n})^{T_{n}}$$
where $h$ is a marked variable $x^{T_{0}}$ or a sort $s$. The symbol
$h$ is called the {\it head} of this term.
\end{definition}

\begin{prop}
Let $t$ be a normal well-typed marked term; $t$ is either an
abstraction, a product or an atomic term.
\end{prop}

\begin{prop}
Let $T$ be a well-typed normal marked term of type $s$ for some sort $s$;
$T$ can be written in a unique way $T = (x_{1}:P_{1}) ... (x_{n}:P_{n})P$ 
with $P$ atomic. Moreover if $s = Type$ then $P = Prop$. 
\end{prop}

\begin{prop} \label{tutu}
Let $\Gamma$ be a marked context and $t$ and $T$ two marked terms such
that $\Gamma \vdashm t:T$. If $t$ is a variable, an application or an 
abstraction, then $T$ is convertible to the outermost mark of $t$. For 
instance if $t = (u~v)^{U}$ then $T$ is convertible to $U$.
\end{prop}

\begin{prop}[Convertibility of contexts]
If $\Gamma [x:T] \Gamma'$ and $\Gamma [x:U] \Gamma'$ are well-formed
contexts and $\Gamma \vdashm T:s$, $\Gamma \vdashm U:s$, $T =_{\beta\eta} U$ 
and $\Gamma [x:T] \Gamma' \vdashm t:V$ hold then 
$\Gamma [x:U] \Gamma' \vdashm t:V$
\end{prop}

\begin{proof}
By induction over the length of the derivation of 
$\Gamma [x:T] \Gamma' \vdashm t:V$.
\end{proof}

\subsection{Subject reduction}

\begin{prop}[Subject $\beta$-reduction]
If $\Gamma \vdashm t:T$ and $t \rhd_{\beta} t'$ then $\Gamma \vdashm
t':T$.
\end{prop}

\begin{proof}
By induction over the structure of $t$. The key case is taken care by the
substitution lemma. For induction loading, we also prove that if 
$\Gamma \rhd_{\beta} \Gamma'$ then $\Gamma' \vdashm t:T$.

\end{proof}

We now want to prove the same property for $\eta$-reduction.

\begin{prop}[Geuvers]
If $\Gamma \vdashm t:u$ and $t =_{\beta\eta} (x:T)U$ then
$t \rhd_{\beta}^{*} (x:V)W$. If $\Gamma \vdashm t:u$, given a sort $s$ such
that $t =_{\beta\eta} s$, one has $t \rhd_{\beta}^{*} s$.
\end{prop}

\begin{proof}
If $t =_{\beta\eta} (x:T)U$ then
$t^{\#} =_{\beta\eta} (x:T^{\#})U^{\#}$. Thus there are 
terms $T'$ and $U'$ such that $t^{\#} \rhd^{*} (x:T')U'$, and 
terms $T''$ and $U''$ such that $t^{\#} \rhd_{\beta}^{*} (x:T'')U''$.
Thus, there are terms $V$ and $W$ such that $t \rhd^{*} (x:V)W$. 

If $t =_{\beta\eta} s$ then
$t^{\#} =_{\beta\eta} s$. Thus $t^{\#} \rhd^{*} s$, 
$t^{\#} \rhd_{\beta}^{*} s$, and $t \rhd_{\beta}^{*} s$. 
\end{proof}

\begin{prop}[Strengthening]
~
\begin{itemize}
\item
If $\Gamma [x:U] \Gamma' \vdashm t:s$ and $x$ occurs free neither in
$\Gamma'$ nor in $t$ then 
$\Gamma \Gamma' \vdashm t:s$.

\item
If $\Gamma [x:U] \Gamma' \vdashm t:T$ and $x$ occurs free neither in
$\Gamma'$ nor in $t$ then there exist a term $T'$ such that 
$\Gamma \Gamma' \vdashm t:T'$.
\end{itemize}
\end{prop}

\begin{proof}
By induction over the structure of $t$:
\begin{itemize}
\item If $t=(t_1~t_2)^A$, by the stripping lemma, we have 
$\Gamma [x:U] \Gamma' \vdashm
t_1:(z:C)D$ and $\Gamma [x:U] \Gamma' \vdashm t_2:C$. The induction
hypothesis ensures that $\Gamma\Gamma'\vdashm t_1:E$ and
$\Gamma\Gamma'\vdashm t_2:C'$. Since $E\conv(z:C)D$, Geuvers' lemma
implies that $E\rhd^{*}_{\beta}(z:C'')D'$. The conversion rule gives
$\Gamma\Gamma'\vdashm t_2:C''$, and so $\Gamma\Gamma'\vdashm
(t_1~t_2)^{D'[z \la t_2]}:D'[z\la t_2]$. The induction
hypothesis also implies $\Gamma\Gamma'\vdashm A:s$, and so we finally
have $\Gamma\Gamma'\vdashm (t_1~t_2)^A:D'[z\la t_2]$.
\item If $t=([z:A]t_0)^B$, the stripping lemma applied to the
derivable judgement 
$$\Gamma[x:U]\Gamma'\vdashm ([z:A]t_0)^{B}:B$$
ensures that
$$\Gamma[x:U]\Gamma'[z:A]\vdashm t_0:C,$$
$$\Gamma[x:U]\Gamma'\vdashm A:s_1,$$
$$\Gamma[x:U]\Gamma'\vdashm B:s_2,$$
$B\conv (z:A)C$, and $<s_1,s_2>\in\Ru$.

Thus, by induction hypothesis 
$$\Gamma\Gamma'[z:A]\vdashm t_0:D,$$
$$\Gamma\Gamma'\vdashm A:s_1,$$
$$\Gamma\Gamma'[z:A]\vdashm B:s_2.$$

Geuvers' lemma ensures that
$B\rhd^{*}_{\beta}(z:A')C'$ with 
$C'\conv D$ and $A\conv A'$. Thus we have
$$\Gamma \Gamma' [z:A] \vdashm t_{0}:C'$$
$$\Gamma \Gamma' [z:A] \vdashm C':s$$
Hence
$$\Gamma \Gamma' \vdashm ([z:A]t_{0})^{(z:A)C'}:(z:A)C'$$
\end{itemize}
The other cases are straightforward.
\end{proof}

\begin{prop}[subject $\eta$-reduction]
If $\Gamma \vdashm t:T$ and $t \rhd_{\eta} t'$ then $\Gamma \vdashm t':T$.
\end{prop}

\begin{proof}
By induction over the structure of $t$. All the cases are straightforward, but
the one in which $t$ is itself the reduced $\eta$-redex. In this case
$t = ([x:U](u~x^{V})^{W})^{X}$ and $t' = u$. Using twice the stripping lemma
we get $\Gamma [x:U] \vdashm u:(y:A)B$ with $(y:A)B \conv T$. By the
strengthening lemma we get $\Gamma \vdashm u:C$ and by unicity of
typing $C \conv (y:A)B$. Since $T \neq Type$, we have $\Gamma \vdashm
T:s$ and  $\Gamma \vdashm u:T$ by conversion.
\end{proof}

\subsection{Normal forms}

We have remarked above that even when $t^{\#}$ is an $\eta$-redex,
$t$ is not necessarily an $\eta$-redex. We prove now that provided we
perform enough reductions inside the marks, the term $t$ becomes an
$\eta$-redex. In the example above, the term 
$$([x:T^{Prop}](y^{(([z:T^{Prop}](T^{Prop} \ra T^{Prop}))^{T^{Prop} \ra
Prop}~x^{T^{Prop}})^{Prop}}~x^{T^{Prop}})^{T^{Prop}}) ^{T^{Prop} \ra
T^{Prop}}$$ is not an $\eta$-redex, but it reduces to the term
$$([x:T^{Prop}](y^{T^{Prop} \ra
T^{Prop}}~x^{T^{Prop}})^{T^{Prop}})^{T^{Prop} \ra T^{Prop}}$$
which is an $\eta$-redex. 

More precisely, we eventually want to prove the following: if $t$ is
well-typed and $\beta \eta$-normal then $t^{\#}$ is $\beta
\eta$-normal.

\begin{prop}
If $\Gamma [x:U] \Gamma' \vdashm t:T$ with $t$ in $\beta$-normal form
and there exists $t'$ such that $t\conv t'$ and 
$$\Gamma\Gamma'\vdashm t':T$$
then $x$ does not occur in $t^{\#}$.
\end{prop}

\begin{proof}
It is an immediate consequence of the proposition~\ref{tec1} and the fact
that
$$\Gamma^{\#}[x:U^{\#}]\Gamma'^{\#}\vdash^{\#} t^{\#}:T^{\#}$$
since $t^{\#}$ is $\beta$-normal.
\end{proof}

\begin{prop} \label{free}
If $\Gamma [x:U] \Gamma' \vdashm t:T$, $t$ is
$\beta$-normal and 
$x$ does not occur free neither in $\Gamma'$ nor in $t^{\#}$ and in $T$,
then $x$ does not occur free in $t$ either (and thus
$\Gamma\Gamma'\vdashm t:T$ holds).
\end{prop}

\begin{proof} 
By induction over the structure of $t$:
\begin{itemize}
\item If $t = (z:A)B$, we have $\Gamma[x:U] \Gamma'\vdashm A:s$ which, by
induction, implies that $x$ does not occur free in $A$. We then apply
the induction hypothesis to $\Gamma[x:U] \Gamma'[z:A]\vdashm B:s'$
which ensures that $x$ does not occur free in $B$.

\item If $t=([z:B]C)^D$, the same reasoning as above ensures that $x$
does not occur free in $B$.
We know that $D \conv T$, that $x$ does not occur free in $T$; since
$D$ is $\beta$-normal, the previous proposition ensures that $x$ does not
occur free in $D^{\#}$. By induction hypothesis, this implies that
$x$ does not occur free in $D$ either.

We know that $D$ (or $T$) is
convertible to a product, and thus $D\rhd^{*}_{\beta}(z:B')E$. We may then
apply the induction hypothesis to $\Gamma[x:U]\Gamma'[z:B]\vdashm C:E$.
Thus the variable $x$ does not occur in $C$.

\item If $t=(~...~(z^{T_{0}}~c_{1})^{T_{1}}~...~c_{n})^{T_{n}}$:
\begin{enumerate}
\item We first prove that $x$ does not occur free in $T_0$. Let $A_0$
be the type bound to $z$ in $\Gamma\Gamma'$. We know that $T_0\conv
A_0$ and hence $T_0^{\#}\conv A_0^{\#}$. We also know that $x$
does not occur free in $A_0$ and thus not in $A_0^{\#}$ either.
Thus $x$ does not occur in the common normal form of $A_{0}^{\#}$ and 
$T_{0}^{\#}$. 
Since $T_0$ is $\beta$-normal, so is $T_0^{\#}$. The previous proposition
allows us to conclude that $x$ does not occur free in $T_0^{\#}$. 
By induction hypothesis,
$x$ does not occur free in $T_0$ (and thus also in $z^{T_0}$).

\item Since $z^{T_0}$ is left-hand of a well-formed application, it
reduces to a function type: $T_0\rhd^{*}_{\beta}(y:A)B$. We know that
$$\Gamma[x:U]\Gamma'\vdashm c_1:A$$
is derivable, that $x$ does not occur free in $c_1^{\#}$  and that $x$
does not occur free in $A$; by 
induction hypothesis, $x$ does not occur free in $c_1$ either.
\item Since $x$ does not occur free in $B$ and $c_1$, it does not
occur free in $B[y \la c_1]$. Furthermore $T_1\conv B[y \la c_1]$ 
and since $T_1$ is $\beta$-normal, the previous proposition ensures 
that $x$ does not occur free in $T_1^{\#}$. Thus, by the induction
hypothesis, it does not occur in $T_1$. 

\item We may then iterate the steps 2 and 3 for every $c_i$ and
$T_i$.
\end{enumerate}
\end{itemize}
\end{proof}

\begin{prop} \label{tyty}
Let  $t$ be a well-typed marked term. If $t$ is $\beta\eta$-normal
then $t^{\#}$ is $\beta \eta$-normal.
\end{prop}

\begin{proof}
We already know that $t^{\#}$ is $\beta$-normal. Now suppose
that $t^{\#}$ is not $\eta$-normal. So we have a subterm $u$ of $t$ such
that $u^{\#}$ is an $\eta$-redex: $u=([x:A](v~x^B)^C)^D$ with $x$ not
occurring free in $v^{\#}$. Let $T$ be the type of $v$. The variable
$x$ does not occur in $v^{\#}$, thus by the unmarked strengthening lemma
there is a type $U$ of $v^{\#}$ in which $x$ does not occur. 
The types 
$U$ and $T^{\#}$
have the same $\beta \eta$-normal form. Thus $x$ does not occur in the 
$\beta \eta$-normal form of $T^{\#}$. Thus $x$ does not occur in the 
$\beta$-normal form of $T^{\#}$ (proposition~\ref{tec1}).
By the proposition~\ref{beta}, $T$ $\beta$-reduces 
to a term $T'$ such that $T'^{\#}$ is the $\beta$-normal form of $T^{\#}$.
We have
$\Gamma [x:T] \vdash T':s$ and $x$ does not occur in $T'^{\#}$, thus,
by the proposition \ref{free}, $x$ does not occur in $T'$.
Then $\Gamma [x:T] \vdash v:T':s$ and $x$ does not occur in $v^{\#}$ 
nor in $T'$, thus, by proposition \ref{free} $x$ does not occur in $v$, 
which is contradictory.
\end{proof}

\subsection{Normalization}

In this section we prove that each well-typed marked term has a normal form.
We use the an adaptation of the method
of \cite{HHP,GeuNed}, i.e. we associate to each marked term $t$ 
an unmarked term $t^{\circ}$ that mimics all its reductions. 

\begin{definition}[Marked terms translation]

We define by simultaneous induction two translations from marked terms to
unmarked terms. We use a fresh variable $o$ (of type $Prop$).
\begin{itemize}
\item
If $P$ is a term of the form 
$(x_{1}:A_{1}) ... (x_{n}:A_{n})Prop$, we let  
$\overline{P} = (x_{1}:A_{1}^{\circ}) ... (x_{n}:A_{n}^{\circ})o$.
\item 
otherwise, we let $\overline{P} = P^{\circ}$.
\end{itemize}
This first auxiliary encoding is meant for types (i.e. terms such that
$t:s$) and in all the cases $\overline{t}:Prop$.

\begin{itemize}
\item
If $t$ is a sort, we let $t^{\circ} =  t$,
\item
if $t = x^{T}$, we let $t^{\circ} = ([z:Prop]x~\overline{T})$,
\item
if $t = (u~v)^{T}$, we let
$t^{\circ} = ([z:Prop](u^{\circ}~v^{\circ})~\overline{T})$,
\item
if $t = ([x:P]u)^{T}$, we let
$t^{\circ} = ([z:Prop][x:P^{\circ}]u^{\circ}~\overline{T})$,
\item
if $t = (x:P)Q$, we let $t^{\circ} = (x:P^{\circ})Q^{\circ}$.
\end{itemize}
\end{definition}

\begin{prop}
\label{FV}
$FV(t^{\circ}) \subseteq FV(t) \cup \{o\}$
\end{prop}

\begin{prop}
$$a^{\circ}[x \la b^{\circ}] \rhd^{*} (a[x \la b])^{\circ}$$
$$\overline{a}[x \la b^{o}] \rhd^{*} \overline{a[x \la b]}$$
\end{prop}

\begin{proof}
By induction over the structure of $a$.
\end{proof}

If $a = x^{T}$ then:
$${a}^{\circ} = ([z:Prop]x~\overline{T})$$
$${a}^{\circ}[x \la b^{\circ}] = ([z:Prop]b^{\circ}~\overline{T}[x \la b^{\circ}])
\rhd^{*} b^{\circ} = (x^{T}[x \la b])^{\circ} = (a[x \la b])^{\circ}$$

The other cases are a simple application of the induction hypothesis.
For instance, if $a$ is an application $a = (t~u)^{T}$, we have
$$a^{\circ} = ([z:Prop](t^{\circ}~u^{\circ})~\overline{T})$$
and
$$a^{\circ}[x \la b^{\circ}] = 
([z:Prop](t^{\circ}[x \la b^{\circ}]~u^{\circ}[x \la b^{\circ}])~
\overline{T}[x \la b^{\circ}])$$
By induction hypothesis we have
\begin{eqnarray*}
t^{\circ}[x \la b^{\circ}] & \rhd^{*} & (t[x \la b])^{\circ}\\
u^{\circ}[x \la b^{\circ}] & \rhd^{*} & (u[x \la b])^{\circ}\\
\overline{T}[x \la b^{\circ}] & \rhd^{*} & \overline{T[x \la b]}
\end{eqnarray*}
So
\begin{eqnarray*}
a^{\circ}[x \la b^{\circ}] & \rhd^{*} &
([z:Prop]((t[x \la b])^{\circ}~(u[x \la b])^{\circ})
\overline{T[x \la b]})\\
& = & ((t[x \la b]~u[x \la b])^{T[x \la b]})^{\circ}
= (a[x \la b])^{\circ}
\end{eqnarray*}

\begin{prop}
If $a \rhd b$ then $a^{\circ} \rhd^{+} b^{\circ}$.
\end{prop}

\begin{proof}
If $a = (([x:P]t)^{T} u)^{U}$ and $b = t[x \la u]$,
$$a^{\circ} 
= ([z:Prop]([z':Prop][x:P^{\circ}]t^{\circ}~\overline{T}~u^{\circ})~\overline{U}^{\circ})$$
$$b^{\circ} = (t[x \la u])^{\circ}$$
By reducing first three $\beta$-redexes in $a^{\circ}$ we get 
$t^{\circ}[x \la u^{\circ}]$. Thus 
$$a^{\circ} \rhd^{+} t^{\circ}[x \la u^{\circ}]$$
and 
$$t^{\circ}[x \la u^{\circ}] \rhd^{*} (t[x \la u])^{\circ} = b^{\circ}$$
thus
$$a^{\circ} \rhd^{+} b^{\circ}$$
If
$a = ([x:P](t~x^{P})^{T})^{U}$ and $b = t$,
$$a^{\circ} 
= ([z:Prop][x:P^{\circ}]([z':Prop](t^{\circ}~([z'':Prop]x~\overline{P}^{\circ}))~\overline{T})~\overline{U}^{\circ})$$
$$b^{\circ} =  t^{\circ}$$
We reduce three $\beta$-redexes and one $\eta$-redex in $a^{\circ}$ to get
$b^{\circ}$ (notice that, by proposition 
\ref{FV}, the variable $x$ does not occur in the term $t^{\circ}$).

The same holds if we reduce a redex in a subterm.
\end{proof}

\begin{prop}
If $a =_{\beta\eta} b$ then $a^{\circ} =_{\beta \eta} b^{\circ}$.
\end{prop}

\begin{proof}
For all terms $t$ we have $t^{\circ} \rhd^{*} t^{\#}$. 
Thus $a^{\circ} =_{\beta \eta} a^{\#} 
=_{\beta \eta} b^{\#} =_{\beta \eta} b^{\circ}$.
\end{proof}

\begin{definition}[Marked context translation]

Let $\Gamma = [x_{1}:P_{1}; ...; x_{n}:P_{n}]$ be a marked context, we let
$\Gamma^{\circ} = [o:Prop; x_{1}:P_{1}^{\circ}; ...; x_{n}:P_{n}^{\circ}]$.
\end{definition}

\begin{prop}
Let $\Gamma$ be a marked context and $t$ and $T$ two marked terms such that 
$\Gamma \vdashm t:T$ in some system of the cube, then
$\Gamma^{\circ} \vdashm t^{\circ}:T^{\circ}$ in the Calculus of
Constructions.
\end{prop}

\begin{proof}
By induction over the length of the derivation of $\Gamma \vdashm t:T$.
\end{proof}

\begin{prop}
The reduction on marked terms is strongly normalizing.
\end{prop}
 
If there was an infinite reduction issued from a marked term $t$, we 
could build one issued from the unmarked term $t^{\circ}$, in contradiction 
with the fact that reduction is strongly normalizing on well-typed terms in 
the Calculus of Constructions.

\subsection{Confluence}

\begin{prop}\label{toto}
Let $a$ and $b$ two $\beta \eta$-normal marked terms well-typed 
in the marked context $\Gamma$. If $a^{\#} = b^{\#}$ then $a = b$.
\end{prop}

\begin{proof}
By induction over the structure of $a$.
Since $a$ and $b$ have the same contents, they are either both sorts,
both variables, both abstractions, both products or both applications.

If they are, for instance, both applications, $a = (t~u)^{T}$, $b = (v~w)^{U}$ 
then the marked terms $t$ and $v$ are well-typed and normal in $\Gamma$ and 
have the same contents so they are equal, 
and symmetrically the marked terms $u$ and $w$ are equal. The marked terms
$T$ and $U$ are well-typed and normal in $\Gamma$. Thus $T^{\#}$ and 
$U^{\#}$ are normal, and both are types of $a^{\#} = b^{\#}$ in
$\Gamma^{\#}$. So $T^{\#} = U^{\#}$ and thus  $T = U$. We then 
conclude that $a = b$.

The same holds if they are both sorts, variables, abstractions or products.
\end{proof}

\begin{prop}[Church-Rosser] Let $\Gamma\vdashm t_1:T$ and
$\Gamma\vdashm t_2:T$ be two derivable judgements. If $t_1\conv t_2$,
then there exists a term $u$ such that $t_1\rhd^{*}u$ and
$t_2\rhd^{*}u$.
\end{prop}
\begin{proof}
Let $u_1$ and $u_2$ be two normal forms of respectively $t_1$ and
$t_2$ (obtained, for instance, by left reduction). 
By proposition \ref{tyty}, the term $u_1^{\#}$
(respectively $u_2^{\#}$) is {\it the} normal form of $t_1^{\#}$
(respectively $t_2^{\#}$). Since $t_1^{\#}\conv t_2^{\#}$, the
Church-Rosser theorem for unmarked terms ensures that $u_1^{\#}=
u_2^{\#}$. The previous proposition then yields $u_1=u_2$. We may then
take $u=u_1=u_2$.
\end{proof}

\begin{prop}[Confluence]
Let $\Gamma\vdashm t:T$,
$t \rhd^{*} t_1$ and $t \rhd^{*} t_2$. Then there exists a term $u$
such that $t_1 \rhd^{*} u$ and $t_2 \rhd^{*} u$.
\end{prop}

\section{Well-foundedness of the relation $<$}

By induction over typing derivations, we associate to each unmarked 
well-typed term a marked well-typed term. 

\begin{definition}
A marked
context $\Gamma = [x_{1}:P_{1}; ... ; x_{p}:P_{n}]$ is said to be normal if 
every $P_{i}$ is normal. 
\end{definition}

\begin{prop}
Let $\Gamma_{1}$ and $\Gamma_{2}$ be two normal 
well-formed marked contexts such that $\Gamma_{1}^{\#} = \Gamma_{2}^{\#}$.
Then $\Gamma_{1} = \Gamma_{2}$. 
\end{prop}

\begin{proof}
By induction over the common length of $\Gamma_{1}$ and $\Gamma_{2}$
using proposition \ref{toto}. 
\end{proof}

\begin{definition}[Translation of an unmarked term into a marked term]

Let us consider an unmarked derivable judgement $\Delta \vdash a:A$ or 
$\vdash \Delta$. By induction over the length 
of this derivation, we build, in the first case, a normal marked
context $\Delta^{*}$ and normal marked terms $a^{*}$ and $A^{*}$ such that 
$\Delta^{*~\#} =_{\beta \eta} \Delta$, $a^{*~\#} =_{\beta \eta}
a$, $A^{*~\#} =_{\beta \eta} A$ and
the judgement $\Delta^{*} \vdashm a^{*}:A^{*}$ is derivable and,
in the second case, a
normal marked context $\Delta^{*}$ such that $\Delta^{*~\#} =_{\beta \eta} \Delta$ and 
the judgement $\vdashm \Delta^{*}$ is derivable.
\begin{itemize}
\item{If the last rule of the derivation is
$$\irule{}
        {\vdash [~]}
        {}$$
we let $\Delta^{*} = [~]$.}
\item{If the last rule of the derivation is
$$\irule{\Gamma \vdash T:s} 
        {\vdash \Gamma[x:T]} 
        {}$$
then by induction hypothesis we have built $\Gamma^{*}$ and $T^{*}$.
We let $\Delta^{*} = \Gamma^{*}[x:T^{*}]$.}
\item{If the last rule of the derivation is 
$$\irule{\vdash \Gamma}
        {\Gamma \vdash Prop:Type} 
        {}$$
then by induction hypothesis we have built $\Gamma^{*}$. We let 
$\Delta^{*} = \Gamma^{*}$, $a^{*} = Prop$ and $A^{*} = Type$.}
\item{If the last rule of the derivation is
$$\irule{\vdash \Gamma~~x:T \in \Gamma} 
        {\Gamma \vdash x:T} 
        {}$$
then by induction hypothesis we have built $\Gamma^{*}$. We have 
$\Gamma^{*~\#} =_{\beta \eta} \Gamma$, so $\Gamma^{*}$ contains a 
declaration $x:P$ and $P^{\#} =_{\beta \eta} T$. We let 
$\Delta^{*} = \Gamma^{*}$, $a^{*} = x^{P}$ and $A^{*} = P$.}
\item{If the last rule of the derivation is
$$\irule{\Gamma \vdash T:s~~\Gamma[x:T] \vdash U:s'}
        {\Gamma \vdash (x:T)U:s'} 
        {}$$
then by induction hypothesis we have built $\Gamma^{*}$, $T^{*}$, 
$(\Gamma[x:T])^{*}$ and $U^{*}$.
Since $(\Gamma[x:T])^{*~\#} =_{\beta \eta} \Gamma[x:T]$, the context
$(\Gamma[x:T])^{*}$ has the form $\Gamma'[x:P]$ with $\Gamma'$ and $P$ normal
$\Gamma'^{\#} =_{\beta \eta} \Gamma$ and $P^{\#} =_{\beta \eta} T$. Then since $\Gamma^{*}$ 
and $\Gamma'$ are both normal well-formed and have the same contents we have 
$\Gamma^{*} = \Gamma'$ and since $T^{*}$ and $P$ are both normal, well-typed 
in $\Gamma^{*}$ and have the same contents $T^{*} = P$. 
We let $\Delta^{*} = \Gamma^{*}$, $a^{*} = (x:T^{*})U^{*}$ and 
$A^{*} = s'$.}
\item{If the last rule of the derivation is
$$\irule{\Gamma \vdash (x:T)U:s~~\Gamma[x:T] \vdash t:U}
        {\Gamma \vdash [x:T]t:(x:T)U} 
        {}$$
then by induction hypothesis we have built $\Gamma^{*}$, $((x:T)U)^{*}$,
$(\Gamma[x:T])^{*}$, $t^{*}$ and $U^{*}$.
Since $((x:T)U)^{*}$ is normal and $((x:T)U)^{*~\#} =_{\beta \eta} (x:T)U$ the term
$((x:T)U)^{*}$ is a product $(x:P)Q$, $P$ and $Q$ are normal 
$P^{\#} =_{\beta \eta} T$ and $Q^{\#} =_{\beta \eta} U$.
In the same way since $(\Gamma[x:T])^{*}$ is normal and 
$(\Gamma[x:T])^{*~\#} =_{\beta \eta} \Gamma[x:T]$, the context 
$(\Gamma[x:T])^{*}$ has the form $\Gamma'[x:R]$ with $\Gamma'$ and $R$ normal
$\Gamma'^{\#} =_{\beta \eta} \Gamma$ and $R^{\#} =_{\beta \eta} U$. 
Since $\Gamma^{*}$ and $\Gamma'$ are normal well-formed and have the same 
contents $\Gamma^{*} = \Gamma'$. Since $P$ and $R$ normal well-typed in
$\Gamma^{*}$ and have the same contents, $P = R$. Since $Q$ and $U^{*}$ are 
normal well-typed in $\Gamma^{*}[x:P]$ and have the same contents, $Q = U^{*}$. 
We let $\Delta^{*} = \Gamma^{*}$, 
$a^{*} = ([x:P]t^{*})^{(x:P)Q}$ and $A^{*} = (x:P)Q$.}
\item{If the last rule of the derivation is
$$\irule{\Gamma \vdash t:(x:T)U~~\Gamma \vdash u:T}
        {\Gamma \vdash (t~u):U[x \la u]}
        {}$$
then by induction hypothesis we have built $\Gamma_{1}^{*}$,
$t^{*}$, $((x:T)U)^{*}$, $\Gamma_{2}^{*}$, $u^{*}$ and $T^{*}$. 
Since the term $((x:T)U)^{*}$ is normal and $((x:T)U)^{*~\#} =_{\beta \eta} (x:T)U$, 
the term $((x:T)U)^{*}$ has the form $(x:P)Q$, $P$ and $Q$ are normal
$P^{\#} =_{\beta \eta} T$ and $Q^{\#} = U$. Since $\Gamma_{1}^{*}$ and $\Gamma_{2}^{*}$
are normal, well-formed and have the same contents,
$\Gamma_{1}^{*} = \Gamma_{2}^{*}$. Since $P$ and $T^{*}$ normal, well-typed
in $\Gamma_{1}^{*}$ and have the same contents, $T^{*} = P$. 
We have $\Gamma_{1}^{*} \vdashm t^{*}:(x:T^{*})U^{*}$ and 
$\Gamma_{1}^{*} \vdashm u^{*}:T^{*}$, thus
$\Gamma_{1}^{*} \vdashm (t^{*}~u^{*})^{U^{*}[x \la u^{*}]}:U^{*}[x \la u^{*}]$.
Let $v$ be the normal form of $(t^{*}~u^{*})^{U^{*}[x \la u^{*}]}$ and 
$V$ be the normal form of $U^{*}[x \la u^{*}]$. We have
$\Gamma_{1}^{*} \vdashm v:V$.
We let $\Delta^{*} = \Gamma_{1}^{*}$,
$a^{*} = v$ and $A^{*} = V$.}
\item{If the last rule of the derivation is
$$\irule{\Gamma \vdash t:T~~\Gamma \vdash U:s~~~T =_{\beta \eta} U}
        {\Gamma \vdash t:U} 
        {}$$
then by induction hypothesis we have built $\Gamma^{*}_{1}$, 
$t^{*}$, $T^{*}$, $\Gamma^{*}_{2}$ and $U^{*}$.
Since $\Gamma^{*}_{1}$ and $\Gamma^{*}_{2}$ are 
normal well-formed and have the same contents they are equal. 
Since $T^{*}$ and $U^{*}$ are normal, well-typed in 
$\Gamma^{*}_{1}$ and have the same contents, they are equal.
We let $\Delta^{*} = \Gamma_{1}^{*}$, $a^{*} = t^{*}$ and
$A^{*} =  U^{*}$.}
\end{itemize}
\end{definition}

\begin{prop}
Let $t$ and $u$ be two unmarked normal terms  
well-typed in two contexts $\Gamma$ and $\Delta$. 
If $t_{\Gamma} < u_{\Delta}$ then $t^{*}$ is a strict subterm of ${u}^{*}$.
\end{prop}

\begin{proof}
We have $\Gamma \vdash t:T$ and $\Delta \vdash u:U$. Thus
$\Gamma^{*} \vdash t^{*}:T^{*}$ and $\Delta^{*} \vdash u^{*}:U^{*}$. 
If $t_{\Gamma}$ is a strict subterm of $u_{\Delta}$ then by induction over the 
structure of $u$, the term $t^{*}$ is a strict subterm of $u^{*}$.
If $t$ is the normal form of the type $U$ of $u$ then 
$t^{*}$ and $U^{*}$ are $\beta \eta$-equivalent and normal, thus by Church-Rosser property, they are equal. 
Hence, using proposition \ref{tutu}, since $t$ is not a sort, 
$t^{*}$ is the outermost mark of $u^{*}$.
\end{proof}

Thus, since the subterm order is well-founded on marked terms, we can deduce:

\begin{theorem}
The relation $<$ is well-founded.
\end{theorem}

Notice that this theorem holds also for a calculus of the cube with 
$\beta$-conversion only. 

\section{The Eta-long normal form}

\begin{definition}[Measure of a Term]

Let $\Gamma$ be a context and $t$ a unmarked normal term well-typed of
type $T$ in $\Gamma$. We define by induction over the order $<$, the
{\it measure} $\mu(t)$ of $t$
\begin{itemize}
\item{If $t$ is a sort then $\mu(t) = 1$,}
\item{If $t = x$ and $T$ then $\mu(t) = \mu(T) + 1$,}
\item{If $t = (u~v)$ then $\mu(t) = \mu(u) + \mu(v) + \mu(T)$,}
\item{If $t = [x:U]v$ then $\mu(t) = \mu(U) + \mu(v) + \mu(T)$,}
\item{If $t = (x:U)V$ then $\mu(t) = \mu(U) + \mu(V)$.}
\end{itemize}
\end{definition}
\medskip

Now we can give the following definition by induction over $\mu(t)$.

\begin{definition}[$\eta$-long form]

Let $\Gamma$ be a context and $t$ a normal term well-typed in $\Gamma$ and 
$(x_{1}:P_{1})...(x_{n}:P_{n})P$ ($P$ atomic) the normal form of its type.
Then
\begin{itemize}
\item
if $t = [x:U]u$ then the $\eta$-long form of $t$ in $\Gamma$ is the term 
$[x:U']u'$ where $U'$ is the $\eta$-long form of $U$ in $\Gamma$ and $u'$ the
$\eta$-long form of $u$ in $\Gamma[x:U]$,
\item
if $t = (x:U)V$ the $\eta$-long form of $t$ in $\Gamma$ is the term $(x:U')V'$
where $U'$ is the $\eta$-long form of $U$ in $\Gamma$ and $V'$ the $\eta$-long 
form of $V$ in $\Gamma[x:U]$,
\item
if $t = (w~c_{1}~...~c_{p})$ the $\eta$-long form of $t$ in $\Gamma$ is the
term 
$[x_{1}:P'_{1}] ... [x_{n}:P'_{n}] (w~c'_{1}~...~c'_{p}~x'_{1}~...~x'_{n})$
where $c'_{i}$ is the $\eta$-long form of $c_{i}$ in $\Gamma$, $P'_{i}$ the 
$\eta$-long form of $P_{i}$ in $\Gamma[x_{1}:P_{1}; ...; x_{i-1}:P_{i-1}]$
and $x'_{i}$ the $\eta$-long form of $x_{i}$ in 
$\Gamma[x_{1}:P_{1}; ...; x_{i}:P_{i}]$.
\end{itemize}
Notice that $\mu(x_{i}) = \mu(P_{i}) + 1 
\leq \mu ((x_{1}:P_{1})...(x_{n}:P_{n})P) < \mu(t)$.
\end{definition}

At last we prove the following extra result. 

\begin{definition}[The Relation $<'$]

Let $<'$ be the smallest transitive relation defined on normal well-typed terms
such that
\begin{itemize}
\item
if neither $t$ nor $u$ is a sort and $t_{\Gamma}$ is a strict subterm of 
$u_{\Delta}$ then $t_{\Gamma} <' u_{\Delta}$,
\item
if neither $t$ nor $T$ is a sort and $T_{\Gamma}$ is the {\it $\eta$-long} 
form of the normal form of the type of $t_{\Gamma}$ in $\Gamma$ and the term 
$t$ is not a sort then $T_{\Gamma} <' t_{\Gamma}$.
\end{itemize}
\end{definition}

We prove that the relation $<'$ is also well-founded.

\begin{definition}[Measure of a Marked Term]

Let $t$ be a marked term, we define by induction over the structure of
$t$, the {\it measure} $\mu(t)$ of $t$
\begin{itemize}
\item{If $t$ is a sort then $\mu(t) = 1$,}
\item{If $t = x^{T}$ then $\mu(t) = \mu(T) + 1$,}
\item{If $t = (u~v)^{T}$ then $\mu(t) = \mu(u) + \mu(v) + \mu(T)$,}
\item{If $t = ([x:U]v)^{T}$ then $\mu(t) = \mu(U) + \mu(v) + \mu(T)$,}
\item{If $t = (x:U)V$ then $\mu(t) = \mu(U) + \mu(V)$.}
\end{itemize}
\end{definition}

\begin{definition}[$\eta$-long Form of a Marked Term]

Let $\Gamma$ be a marked context and $t$ be a normal marked term well-typed in
$\Gamma$ with the type
$(x_{1}:P_{1})...(x_{n}:P_{n})P \quad (P ~{\rm atomic})$.
The {\it $\eta$-long form} of the marked
term $t$ is defined by induction over $\mu(t)$.
\begin{itemize}
\item
If $t = ([x:U]u)^{T}$ then the $\eta$-long form of $t$ is $([x:U']u')^{T'}$
where $U'$ is the $\eta$-long form of $U$ in $\Gamma$, $T'$ the $\eta$-long 
form of $T$ in $\Gamma$ and $u'$ the $\eta$-long form of $u$ in $\Gamma[x:U]$,
\item
if $t = (x:U)V$ then the $\eta$-long form of $t$ is $(x:U')V'$ where
$U'$ is the $\eta$-long form of $U$ in $\Gamma$ and $V'$ the $\eta$-long form 
of $V$ in $\Gamma[x:U]$,
\item
if $t = (~...~(w^{T_{0}}~c_{1})^{T_{1}}~...~c_{p})^{T_{p}}$
then the $\eta$-long form of $t$ is 
$$[x_{1}:P'_{1}] ... [x_{n}:P'_{n}]
(~...~((~...~(w^{T'_{0}}~c'_{1})^{T'_{1}}~...~c'_{p})
^{T'_{p}}~x'_{1})^{V'_{1}}~...~x'_{n})^{V'_{n}}$$
where $c'_{i}$ is the $\eta$-long form of $c_{i}$ in $\Gamma$,
$T'_{i}$ the $\eta$-long form of $T_{i}$ in $\Gamma$,
$P'_{i}$ the $\eta$-long form of $P_{i}$ in 
$\Gamma[x_{1}:P_{1}; ...; x_{i-1}:P_{i-1}]$,
$P'$ the $\eta$-long form of $P$ in 
$\Gamma[x_{1}:P_{1}; ...; x_{n}:P_{n}]$,
$x'_{i}$ the $\eta$-long form of $x_{i}^{P_{i}}$ in
$\Gamma[x_{1}:P_{1}; ...; x_{i}:P_{i}]$
and 
$V'_{i} = (x_{i+1}:P'_{1})...(x_{i+1}:P'_{n})P'$.
\end{itemize}
\end{definition}

\begin{definition}[Normal $\eta$-long translation of an unmarked term]

Let $t$ be a term well-typed in the context $\Gamma$, we define its {\it normal
$\eta$-long translation} $t^{+}$ as the $\eta$-long form of its 
translation $t^{*}$.
\end{definition}

\begin{theorem}
The relation $<'$ is well-founded.
\end{theorem}

\begin{proof}
As in the proof of the well-foundedness of the relation $<$
we first prove that if $t <' u$ then $t^{+}$ is a strict subterm of 
$u^{+}$ and then that there is no infinite decreasing sequence for
the relation $<'$.
\end{proof}
 
\section*{Acknowledgements}

The authors thank Christine Paulin who suggested them the idea of the 
normalization proof for marked terms and Cristina Cornes for a careful 
reading of a draft of this paper.

\end{document}